\documentclass{svjour3}

\usepackage[utf8]{inputenc}
\usepackage{microtype}

\usepackage{amsmath}
\usepackage{graphicx}

\usepackage[ruled,linesnumbered,noend]{algorithm2e}

\newcommand{\hn}{\ensuremath{\mathcal{H}}}

\renewcommand{\epsilon}{\varepsilon}
\renewcommand{\tilde}{\widetilde}
\newcommand{\hide}[1]{\relax}

\newtheorem{fact}{Fact}

\title{Fault-Tolerant Approximate \\Shortest-Path Trees\thanks{A preliminary version of this paper appeared on the \emph{Proceedings of the 22nd European Symposium on Algorithms (ESA’14)}, September 8-10, 2014, Wroclaw, Poland, Vol. 8737 of Lecture Notes in Computer Science, Springer, pp. 137--148.
This work was partially supported by the Research Grant PRIN 2010 ``ARS TechnoMedia'', funded by the Italian Ministry of Education, University, and Research.}}

\author{Davide~Bil\`o
\and
Luciano~Gual\`a
\and
Stefano~Leucci
\and
Guido~Proietti}

\institute{Davide Bil\`o \at Dipartimento di Scienze Umanistiche e Sociali, Università di Sassari, Italy. \\ \email{davide.bilo@uniss.it} \and
Luciano Gual\`a \at Dipartimento di Ingegneria dell'Impresa, Università di Roma ``Tor Vergata'', Italy. \\ \email{guala@mat.uniroma2.it} \and
Stefano Leucci \at Dipartimento di Informatica, ``Sapienza'' Università di Roma, Italy. \\ \email{leucci@di.uniroma1.it} \and
Guido Proietti \at Dipartimento di Ingegneria e Scienze dell'Informazione e Matematica, Università degli Studi dell'Aquila, Italy and Istituto di Analisi dei Sistemi ed Informatica, CNR, Roma, Italy. \\ \email{guido.proietti@univaq.it}}

\smartqed

\def\makeheadbox{\relax}
\date{}

\begin{document}

\maketitle

\begin{abstract}
The resiliency of a network is its ability to remain \emph{effectively} functioning also when any of its nodes or links fails.
However, to reduce operational and set-up costs, a network should be small in size, and this conflicts with the requirement of being resilient.
In this paper we address this trade-off for the prominent case of the {\em broadcasting} routing scheme, and we build efficient (i.e., sparse and fast) \emph{fault-tolerant approximate shortest-path trees}, for both the edge and vertex \emph{single-failure} case. In particular, for an $n$-vertex non-negatively weighted graph, and for any constant $\varepsilon >0$, we design two structures of size $O(\frac{n \log n}{\varepsilon^2})$ which guarantee $(1+\varepsilon)$-stretched paths from the selected source also in the presence of an edge/vertex failure. This favorably compares with the currently best known solutions, which are for the edge-failure case of size $O(n)$ and stretch factor 3, and for the vertex-failure case of size $O(n \log n)$ and stretch factor $3$. Moreover, we also focus on the unweighted case, and we prove that an ordinary \emph{spanner} can be slightly augmented in order to build efficient fault-tolerant approximate \emph{breadth-first-search trees}.
\end{abstract}


\section{Introduction}
Broadcasting a message from a source node to every other node of a network is one of the most basic communication primitives. Since this operation should be performed by making use of a both sparse and fast infrastructure, the natural solution is to root at the source node a {\em shortest-path tree} (SPT) of the underlying graph. However, the SPT, as any tree-based network topology, is highly sensitive to a link/node malfunctioning, which will unavoidably cause the disconnection of a subset of nodes from the source.

To be readily prepared to react to any possible (transient) failure in a SPT, one has then to enrich the tree by adding to it a set of edges selected from the underlying graph, in order to obtain a subgraph that approximately preserves the distance from the source vertex even when a single component (i.e., edge or vertex) fails.
More formally, if $s$ denotes a distinguished source vertex of an undirected graph $G=(V(G),E(G))$ with non-negative real weights on its edges, we say that a spanning subgraph $H$ of $G$ is an \emph{Edge-fault-tolerant $\alpha$-Approximate SPT}  (in short, $\alpha$-{\ttfamily EASPT}), with $\alpha>1$, if it satisfies the following condition: For each edge $e \in E(G)$, all the distances from $s$ in the subgraph $H-e=(V(H), E(H) \setminus \{e\})$ are at most $\alpha$ times longer than the corresponding distances in $G-e$.
When \emph{vertex failures} are considered, then the {\ttfamily EASPT} is correspondingly called {\ttfamily VASPT}.
Ideally we would like a {\ttfamily E/VASPT} to have both a low \emph{stretch} $\alpha$ and a small \emph{size}, measured as the number $|E(H)|$ of edges in $H$.
The case in which $\alpha=1$ corresponds to requiring all the post-failure distances in $H-e$ to match the distances in $G-e$, i.e., $H$ must contain a SPT (from $s$) of $G - e$ for every $e \in E(G)$. However, in this case, it is easy to see that $\Omega(n^2)$ edges might be required, as shown in Figure~\ref{fig:n2_edges_example_weighted}.
%It is easy to see that, in this case, up to $\Omega(n^2)$ edges might be required in any $1$-EASPT of $G$ as it is shown, e.g., in Figure~\ref{fig:n2_edges_example_weighted}.

%structure will be 2-edge/vertex-connected  w.r.t.\ the source. Thus, after an edge/vertex failure,
%these edges will be used to build up the alternative paths emanating from the root, each one of them in replacement of a corresponding original shortest path which was affected by the failure. However, if these paths are demanded to be \emph{shortest}, then it can be easily seen that for a non-negatively real weighted and undirected graph $G$ of $n$ nodes and $m$ edges, this may require as much as $\Theta(m)$ additional edges, also in the case in which $m=\Theta(n^2)$, as we will show later. In other words, the set-up costs of the strengthened network may become unaffordable.
%Thus, a reasonable compromise is that of building a \emph{sparse} and \emph{fault-tolerant} structure which \emph{approximates} the shortest paths from the source, i.e., that contains paths which are longer than the corresponding shortest paths by at most a multiplicative \emph{stretch} factor, for any possible edge/vertex failure.

The aim of this paper is to show that, as soon as we allow for approximate distances, we can obtain an almost optimal stretch-size tradeoff for {\ttfamily E/VASPT}s.

\subsection{Related work}

A problem that is very closely related to the design of a {\ttfamily E/VASPT} is that of computing a \emph{single-source distance sensitivity oracle} ({\ttfamily SDSO}). Designing an efficient {\ttfamily SDSO} means to compute, with a \emph{low} preprocessing time, a \emph{compact} data structure which is able to \emph{quickly} return a (possibly \emph{approximate}) distance between a source vertex $s$ and any other vertex of the graph, following a component failure.
Notice that any {\ttfamily E/VASPT} $H$ also implies the existence of a (trivial) {\ttfamily SDSO} having the same size, the same stretch, and a query time of $O(|E(H)| + n \log n)$: this {\ttfamily SDSO} is obtained by storing the whole graph $H$ and by running Dijkstra's algorithm from $s$ on the surviving graph to answer queries.

\begin{figure}
	\centering
	\includegraphics[scale=1.05]{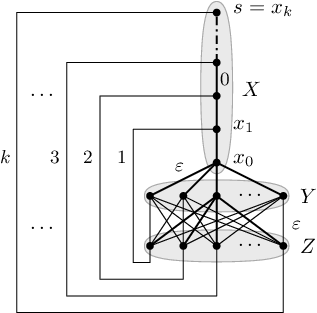}
	\caption{A weighted graph $G$ on $n$ vertices and $\Theta(n^2)$ edges such that an edge-fault-tolerant SPT $H$ of $G$ is $G$ itself. The vertices are partitioned into three sets $X=\{x_0, \dots, x_k = s\}$, $Y=\{y_1, \dots, x_k\}$, and $Z=\{z_1, \dots, z_k\}$. Vertices in $X$ are connected by a path whose edges have weight 0, there is a star centered in $x_0$ whose leaves are the vertices in $Y$, and the sets $Y$ and $Z$ induce a complete bipartite graph. All edge weights of the star and the complete bipartite graph are equal to $\varepsilon$ with $0 < \varepsilon < \frac{1}{2}$, while, for $0<i \le k$, there is an edge of cost $i$ between $x_i$ and $z_i$. A SPT of $G$ from $s$ is shown with bold edges. If any edge $e=(x_i, x_{i+1})$ fails the new SPT of $G-e$ must include all the edges connecting $z_i$ to the vertices in $Y$.}
	\label{fig:n2_edges_example_weighted}
\end{figure}

In \cite{BK13} the authors compute in $O(m \log n + n^2 \log n)$ time a {\ttfamily SDSO} of size $O(n \log n)$, which reports, in constant time per query, $3$-stretched distances following the failure of a single vertex.
Such an oracle is also \emph{path-reporting}, i.e., it is able to return the path associated with a distance query by paying an additional time which is proportional to the number of edges it contains.
A closer inspection of this result shows that this {\ttfamily SDSO} is actually obtained through the computation of a $3$-{\ttfamily VASPT} of size $O(n \log n)$.
Regarding single edge failures, in \cite{NPW03} are (implicitly) provided (i) a path-reporting {\ttfamily SDSO} having stretch $3$, size $O(n)$, and constant query time, and (ii) a corresponding $3$-{\ttfamily EABFS}\footnote{We use the notation {\ttfamily E/VABFS} instead of {\ttfamily E/VASPT} to stress the fact that we are dealing with unweighted graphs.} containing $2(n-1)$ edges.
Very recently, in \cite{BGLP16ESA}, the authors show how to build a (non path-reporting) {\ttfamily SDSO} having stretch $1+\epsilon$, size $O(n \delta)$ and query time $O(\delta \log n)$, where $\delta = \epsilon^{-1} \log \epsilon^{-1}$, which can be improved to $O(1)$ for the special case $\epsilon=1$.

If we focus on unweighted graphs and we insist on preserving exact distances (i.e., stretch equal to $1$) then, in \cite{PP13}, the authors provide a $1$-{\ttfamily E/VABFS} of size $O(n \cdot \min\{\mathit{ecc}(s),\sqrt{n}\})$, where $\mathit{ecc}(s)$ denotes the eccentricity of the source $s$ in $G$. In the same paper, the authors also exhibit a corresponding lower bound of $\Omega(n^{3/2})$ for the size of such a structure (in fact, the construction provided in Figure~\ref{fig:n2_edges_example_weighted} is obtained by elaborating such lower bound).
In \cite{BK13} the authors focus on the vertex-failure case and, for any $\varepsilon>0$, they compute in $O(m \sqrt{n/\varepsilon})$ time a path-reporting {\ttfamily SDSO} of size $O(\frac{n}{\varepsilon^3}+n \log n)$, stretch $(1+\varepsilon)$, and having constant query time.
Once again, this {\ttfamily SDSO} is obtained through the construction of a $(1+\varepsilon)$-{\ttfamily VABFS} of size $O(\frac{n}{\varepsilon^3}+n \log n)$.
Actually, we point out that the latter structure can be easily sparsified so as to obtain, for any $\varepsilon>0$, a $(1+\varepsilon)$-{\ttfamily EABFS} of size $O(\frac{n}{\varepsilon^3})$: indeed, its $O(n \log n)$-size term is associated with an auxiliary substructure that, for the case of edge failures, can be made of linear size. This result is of independent interest, since it qualifies itself as the best current solution for the {\ttfamily EABFS} problem.
In \cite{PP14} the authors present, among other results, a $3$-{\ttfamily EABFS} having at most $4n$ edges. Interestingly, this was the first \emph{explicit} construction for the problem, but two (better) implicit solutions were already available in the literature: the first one is the just mentioned structure which can be derived from the results presented in \cite{BK13}, while the second one is the $3$-{\ttfamily EASPT} of size at most $2n$ (and then, \emph{a fortiori}, a $3$-{\ttfamily EABFS} of the same size) of \cite{NPW03} that can be easily obtained as a by-product of the results given therein (we will discuss this point in more detail later).

\subsection{Our results} Our main result is a polynomial time construction\footnote{We do not insist on the time efficiency in building our structures, since the focus of our paper, consistently with the literature, is on the trade-off between their size and their stretch factor.} of a $(1+\varepsilon)$-{\ttfamily E/VASPT} of size $O(\frac{n \log n}{\varepsilon^2})$, for any $\varepsilon >0$. These two structures substantially improve the stretch of the $3$-{\ttfamily EASPT} of linear size implicitly given in \cite{NPW03}, and that of the $3$-{\ttfamily VASPT} of size $O(n \log n)$ given in \cite{BK13}, respectively, while essentially using the same number of edges (up to a logarithmic factor in the former case). To obtain our results, we perform
a careful selection of edges that will be added to an initial SPT.
The somewhat surprising outcome of our approach is that if we accept to have slightly stretched fault-tolerant paths, then we can drastically reduce the $\Theta(n^2)$ size of the structure that we would have to pay for having fault-tolerant \emph{shortest} paths!
Actually, the analysis of the stretch factor and of the structures' size induced by our algorithms is quite involved. Thus, for clarity of presentation, we give our result in two steps: first, we show an approach to build a $(1+\varepsilon)$-{\ttfamily EASPT} of size $O(\frac{n \log n}{\varepsilon^2})$, then we outline how this approach can be extended to the vertex-failure case.

We also focus on the unweighted case, and we exhibit an interesting connection between a fault-tolerant approximate BFS and an \emph{$(\alpha,\beta)$-spanner}. An $(\alpha,\beta)$-spanner of a graph $G$ is a spanning subgraph $H$ of $G$ such that \emph{all} the node-to-node distances in $H$ are stretched by at most a multiplicative factor of $\alpha$ plus an additive term of $\beta$ w.r.t.\ the corresponding distances in $G$.
If such a condition holds even after an edge/vertex is deleted from both $G$ and $H$, then $H$ is an \emph{edge/vertex-fault-tolerant $(\alpha,\beta)$-spanner}.
Moreover, if the guarantee on the stretch only holds for distances from vertices in a subset $S \subseteq V(G)$, then the spanner is said to be \emph{sourcewise}.
We show how a $(\alpha,\beta)$-spanner of size $\sigma=\sigma(n,m)$ can be used to build in polynomial time a sourcewise edge-fault-tolerant (resp. vertex-fault-tolerant) $(\alpha,\beta)$-spanner of size $O(\sigma + |S| \cdot n)$ (resp., $O(\sigma + |S| \cdot n\log n)$).
This result has three main consequences.
First of all notice that when $|S|=1$, a sourcewise edge/vertex-fault-tolerant $(\alpha,0)$-spanner is exactly an $\alpha$-{\ttfamily E/VABFS}. As a consequence, 
for relevant values of $\alpha$ and $\beta$ (e.g., when they are constant) the {\ttfamily E/VABFS} problem is easier than the corresponding (non fault-tolerant) spanner problem, and we regard this as an interesting hardness characterization.\footnote{For constant values of $\alpha$ and $\beta$, the size of an $(\alpha,\beta)$-spanner is $\omega(n \log n)$ and hence the additive terms in the size of our {\ttfamily E/VABFS} are dominated by $\sigma$.}
A second consequence, is that this bridge between the two problems allows to build the sparsest $(1,\beta)$-{\ttfamily VABFS} structures known so far, by making use of the vast literature on additive $(1,\beta)$-spanners. More precisely, the $(1,4)$-spanner of size $\widetilde{O}(n^\frac{7}{5})$ given in \cite{Che13}, and the $(1,6)$-spanner of size $O(n^\frac{4}{3})$ given in \cite{BKMP10}, can be used to build corresponding {\ttfamily VABFS} structures.
As a last consequence of our result, we are able to: (i) sparsify, for $|S| = \widetilde{\omega}(n^\frac{1}{15})$, the sourcewise edge-fault-tolerant $(1,4)$-spanner of size $O(|S| \cdot n^{\frac{4}{3}})$ given in \cite{PP14} by reducing its size to $\widetilde{O}(n^\frac{7}{5} + |S| n)$; and (ii) reduce the stretch of the sourcewise vertex-fault-tolerant $(1,8)$-spanner of size $\widetilde{O}(n^\frac{4}{3})$ given in \cite{P14} to $(1,6)$, for $|S|=\widetilde{O}(n^\frac{1}{3})$ (see Section~\ref{sct:unweighted} for the exact bounds of the obtained spanners).

%Another interesting implication of our result arises for the design of \emph{sourcewise edge/vertex-fault-tolerant spanners} ({\ttfamily SES/SVS} in the following) in unweighted graphs, where one wants to find a sparse subgraph which tolerates failures and approximately preserves all distances from a given set $S \subseteq V$ of sources. Indeed, given an $(\alpha,\beta)$-spanner of size $\sigma$, one can build an $(\alpha,\beta)$-{\ttfamily SES} of size $O(\sigma+|S| \cdot n)$, by simply iterating the above construction for each source.

\subsection{Other related results}

\paragraph{Additive \texttt{EABFS} structures.}
In addition to the already cited results, in \cite{PP14} the authors also consider $(\alpha,\beta)$-{\ttfamily EABFS}, i.e., edge-fault-tolerant structures for which the length of a path is stretched by at most a factor of $\alpha$ plus an additive term of $\beta$.
In particular, they prove that $(1,3)$-{\ttfamily EABFS} structures admit a lower bound of
$\Omega(n^{5/4})$ edges, thus showing an interesting dichotomy between multiplicative and additive
stretches, i.e., the fact that additive stretches require super-linear size. Moreover, they construct a $(1,4)$-{\ttfamily EABFS} of size $O(n^{4/3})$.

\paragraph{Sourcewise \texttt{E/VABFS} structures.}
In \cite{PP13}, the same authors extend the already cited $1$-{\ttfamily E/VABFS} of size $O(n^\frac{3}{2})$ to the
sourcewise case, i.e., that in which the structure incorporates an edge-fault-tolerant BFS rooted at each vertex of a set $S \subseteq V(G)$. Here, they show the existence of a solution of size $O(\sqrt{|S|} \cdot n^{3/2})$, which is tight.
Moreover, they also consider the optimization problem of constructing a minimum-size sourcewise 1-\texttt{E/VABFS}, and they provide a corresponding tight $O(\log n)$-approximation algorithm.

\paragraph{Multiple edge failures.}
Regarding multiple edge failures, Parter in \cite{P15} presented a 2-edge-fault-tolerant \emph{exact} BFS having $O(n^{5/3})$ edges, which is tight, while in \cite{PP14} it is shown the existence of a $(3(f +1),(f +1) \log n)$-{\ttfamily EABFS} of size $O(fn)$ for any number $f=O(1)$ of failed edges. This latter result has been improved in \cite{BGLP16} where the authors prove the existence of a $(2|F| + 1)$-{\ttfamily EASPT} of size $O(fn)$ which tolerates the failure of any set $F$ of edges of size at most $f$. This structure can be converted into a corresponding {\ttfamily SDSO} having the same size, and with query time $O(|F|^2 \log^2 n)$.
Moreover, if one is willing to use $O(m \log^2 n)$ space, such an oracle is also able to handle any number of edge failures (i.e., up to $m$).
In \cite{DDFLP15}, the special case of \emph{shortest-path failures} was considered, where the set of failing edges $F$ is supposed to form a source-leaf subpath in a given SPT of $G$. In particular, for the case $|F|=2$, they give an {\ttfamily SDSO} achieving stretch $3$, size $O(n \log n)$, and constant query time.
%here the edges in $F$ must lie along any source-leaf path in a SPT of $G$. For this problem, the authors build both a $(2k - 1)(2|F| + 1)$-{\ttfamily EASPT} of size $O(k n f^{1+1/k})$, where $k \ge 1$ is a parameter of choice, and a corresponding {\ttfamily SDSO} having constant query time. Finally, for the special case of $f=2$, they give an ad-hoc {\ttfamily SDSO} construction achieving stretch $3$, size $O(n \log n)$ and constant query time.

\paragraph{Directed graphs.}

For single-source distances on directed graphs with integer positive edge weights bounded by $M$, in \cite{GW12} it is shown how to build efficiently in $\widetilde{O}(M n^{\omega})$ time, where $\omega< 2.373$ denotes the matrix
multiplication exponent, a randomized edge-fault-tolerant {\ttfamily SDSO} of size $\Theta(n^2)$ returning in $O(1)$ time distances from the source which are exact w.h.p.

\paragraph{Fault-tolerant spanners.}

Another setting which is very close in spirit to ours is that of \emph{fault-tolerant spanners}.
In \cite{CLPR09}, for weighted graphs and any integer $k \geq 1$, the authors present a $(2k-1,0)$-spanner resilient to $f$ vertex (resp., edge) failures of size $O(f^2 \, k^{f+1} \, n^{1+1/k} \, \log^{1-1/k} n)$ (resp., $O(f \, n^{1+1/k})$). This was later improved through a randomized construction in \cite{DK11}.
For a comparison, the sparsest known $(2k-1)$-multiplicative ordinary spanner has size $O(n^{1+1/k})$ \cite{ADDJS93}, and this is believed to be asymptotically tight due to the long-standing girth conjecture of Erd\H{o}s \cite{E64}.
Finally, we mention that in \cite{AFIR13} it was introduced the resembling concept of \emph{resilient spanners}, i.e., spanners such that whenever any edge in $G$ fails, then the relative distance increases in the spanner are very close to those in $G$, and it was shown how to build a resilient spanner by augmenting an ordinary spanner.

Concerning unweighted graphs, it makes instead sense to study fault-tolerant additive spanners. In particular, Braunshvig et al. \cite{BCP12} proposed the following general approach to build an additive spanner tolerating up to $f$ edge failures: Let $A$ be an $f$-edge-fault-tolerant $(\alpha,0)$-spanner, and let $B$ be an ordinary $(1,\beta)$-spanner. Then $H=A\cup B$ is an $f$--edge-fault-tolerant $(1,2f(2\beta+\alpha-1)+\beta)$-spanner. Recently, in \cite{BGGLP15} the corresponding analysis has been refined yielding a better additive bound of $2f(\beta + \alpha- 1) + \beta$, and, more in general, improved fault-tolerant additive spanners have been presented.
Also very close to our present work are the (non-fault-tolerant) \emph{sourcewise spanners} (which, again, approximately preserves all distances from a given set $S \subseteq V$ of sources). In that respect, in \cite{CGK13} the authors give, for any $k >1$, a structure with additive stretch  $2k$ and size $O(n^{1+1/(2k+1)}(k|S|)^{k/(2k+1)})$, which in particular for $k=\log n$ returns a structure with additive stretch $2 \log n$ and size $O(n \sqrt{|S| \log n})$. To the best of our knowledge, no results are instead known for the weighted case.

\paragraph{Further related works.}
For recent achievements on all-to-all distance sensitivity oracles, we refer the reader to \cite{BK13,BK09,CLPR10,DP09}, while for other results on single-edge/vertex failures spanners/oracles on unweighted graphs, we finally refer the reader to \cite{ARFI12,BK13,P14,P16}.

\subsection{Paper organization}
The paper is organized as follows: in Section~\ref{sct:notation} we introduce the notation that will be used throughout the paper; in Section~\ref{sct:3EASPT} we revisit one of the swap procedures presented in \cite{NPW03} to formally prove that it can be used to build a simple 3-{\ttfamily EASPT}; in Section~\ref{sct:EASPT} and \ref{sct:VASPT} we present our main results, namely a $(1+\varepsilon)$-{\ttfamily EASPT} and a $(1+\varepsilon)$-{\ttfamily VASPT}, respectively; in Section~\ref{sct:unweighted} we focus on unweighted graphs, and we show the connection between an {\ttfamily E/VABFS} and an $(\alpha,\beta)$-spanner; finally, in Section~\ref{sct:conclusions} we conclude the paper by outlining few directions for future research.

\section{Notation}
\label{sct:notation}
We start by introducing our notation. For the sake of brevity, we give it for the case of edge failures, but it can be naturally extended to the node failure case.

Given a non-negatively real weighted, undirected graph $G$, we will denote by $w_G(e)$ or $w_G(u,v)$ the weight of the edge $e=(u,v) \in E(G)$.
We also define $w(G) = \sum_{e \in E(G)} w(e)$.
Given an edge $e=(u,v)$, we denote by $G - e$ or $G - (u,v)$ (resp., $G + e$ or $G + (u,v)$) the graph obtained from $G$ by removing (resp., adding) the edge $e$.
Similarly, for a set $F$ of edges, $G-F$ (resp., $G + F$) will denote the graph obtained from $G$ by removing (resp., adding) the edges in $F$.

We will call $\pi_G(x,y)$ a shortest path between two vertices $x,y \in V(G)$, $d_G(x,y)$ its (weighted) length, and $T_G(s)$ a SPT of $G$ rooted at $s$. Whenever the graph $G$ and/or the vertex $s$ are clear from the context, we might omit them, i.e., we will  write $\pi(u)$ and $d(u)$ instead of $\pi_G(s,u)$ and $d_G(s,u)$, respectively. When considering an edge $(x,y)$ of an SPT we will assume $x$ and $y$ to be the closest and the furthest endpoints from $s$, respectively.

Given an edge $e \in E(G)$, we define $\pi_G^{-e}(x,y)$, $d_G^{-e}(x,y)$ and $T_G^{-e}(s)$ to be, respectively, a shortest path between $x$ and $y$, its length, and a SPT in the graph $G-e$. Moreover, if $P$ is a path from $x$ to $y$ and $Q$ is a path from $y$ to $z$, with $x,y,z \in V(G)$, we will denote by $P \circ Q$ the path from $x$ to $z$ obtained by concatenating $P$ and $Q$.

Given $G$, a vertex $s \in V(G)$, and an edge $e=(u,v) \in E(T_G(s))$,
we denote by $U_{G}(e)$ and $D_{G}(e)$ the partition of $V(G)$ induced by the two connected components of $T(G)-e$, such that $U_{G}(e)$ contains $s$ and $u$, and $D_{G}(e)$ contains $v$.
Then, $C_{G}(e) = \{ (x,y) \in E(G) \, : \, x \in U_{G}(e), y \in D_{G}(e) \}$ will denote the \emph{cutset} of $e$, i.e., the set of edges crossing the cut $(U_{G}(e), D_{G}(e))$.

For the sake of simplicity we consider only edge weights that are strictly positive. However our entire analysis also extends to non-negative weights. We also assume, w.l.o.g., that the input graph $G$ is 2-edge/vertex-connected, to avoid pathological failures that would disconnect the graph.
Throughout the rest of the paper we will assume that, when multiple shortest paths exist, ties will be broken in a consistent manner.
In particular we fix a SPT $T=T_G(s)$ of $G$ and, given a graph $H \subseteq G$ and $x,y \in V(H)$, whenever we compute the path $\pi_H(x,y)$ and ties arise, we will prefer the edges in $E(T)$.\footnote{The notation $H \subseteq G$ means that $H$ is a subgraph of $G$.}
We will also assume that if we are considering a shortest path $\pi_H(x,y)$ between $x$ and $y$ passing through vertices $x'$ and $y'$, then $\pi_H(x',y') \subseteq \pi_H(x,y)$.

\section{A 3-{\ttfamily EASPT} structure with at most $2n$ edges}
\label{sct:3EASPT}
We here provide a revisitation of one of the swap procedures presented in \cite{NPW03}  to formally prove that it can be used to build a simple 3-{\ttfamily EASPT}  with at most $2n$ edges, on which our construction of the $(1+\varepsilon)$-{\ttfamily EASPT} will rely.
More precisely, in \cite{NPW03} the authors were concerned with the problem of reconnecting in a best possible way (w.r.t.\ a set of distance criteria) the two subtrees of an SPT undergoing an edge failure, through a careful selection of a \emph{swap edge}, i.e., an edge with an endvertex in each of the two subtrees. In particular, they show that if we select as a swap edge for $e=(u,v)$ -- with $u$ closer to the source $s$ than $v$ -- the edge that lies on a shortest path in $G-e$ from $s$ to $v$,
%(or, alternatively, that one which would be firstly added by a Dijkstra computation of the new shortest paths towards the disconnected nodes),
then the distances from the source towards all the disconnected vertices is stretched at most by a factor of 3.\footnote{Actually, in \cite{NPW03} it is not explicitly claimed the 3-stretch factor, but this is implicitly obtained by the qualitative analysis of the swap procedure therein provided.}
%For the sake of self-containment, we provide in Appendix 1 a formal proof of such result.
Therefore, a $3$-{\ttfamily EASPT} of size at most $2n$ can be obtained by simply adding to a SPT rooted at $s$ such a swap edge for each corresponding tree edge, and interestingly this improves the $3$-{\ttfamily EASPT} of size at most $4n$ provided in \cite{PP14}.

\begin{algorithm}[t]
	\DontPrintSemicolon
	\SetKwInOut{Input}{Input}
	\SetKwInOut{Output}{Output}
	
	\Input{A non-negatively real weighted graph $G$, $s \in V(G)$}
	\Output{A 3-{\ttfamily EASPT} of $G$ rooted at $s$}

	\BlankLine	
		
	$H \gets T_G(s)$
	
	\For{$e=(u,v) \in E(T_G(s))$}
	{
		$f_e \gets$ the single edge in $C_{G}(e) \cap \pi_G^{-e}(v)$ \;
		$H \gets H + f_e$ \;
	}
	\Return $H$
	
	\caption{Algorithm for building a 3-{\ttfamily EASPT}}\label{alg:3ftspt}
\end{algorithm}

More formally, Algorithm~\ref{alg:3ftspt} builds a $3$-\texttt{EASPT} $H$ as follows: initially $H$ is a shortest path tree of $G$ then, for each possible failure of an edge $e=(u,v)$ in $T_G(s)$, we augment $H$ by adding the (unique) edge $(x,y)$ of $C_{G}(e)$ that lies on a shortest path $\pi_G^{-e}(v)$ from $s$ to $v$. Notice that this is the only edge of $\pi_G^{-e}(v)$ that is not already in $T$ as both $\pi_T(s,x)$ and $\pi_T(v,y)$ do not contain $e$.

\begin{lemma}
	Algorithm~\ref{alg:3ftspt} computes in polynomial time a 3-{\ttfamily EASPT} structure of size $2n-2$.
\end{lemma}
\begin{proof}
	The claim on the size of $H$ is a direct consequence of the fact that we add at most one replacement edge for each failure, so we only need to prove that $H$ is a 3-{\ttfamily EASPT} structure.
	As $H$ contains all the edges of $T_G(s)$ the condition $d_H(u)=d_G(u) \; \forall u \in V(G)$ is clearly true. Moreover, the above still holds whenever an edge $e \not\in E(T_G(s))$ fails.
	
	Now, let $e=(u,v) \in E(T_G(s))$ be the failed edge and let $t$ be any vertex in $V(G)$.
	If $t$ belongs to $U_{G}(e)$ then $H$ contains the whole shortest path $\pi_G^{-e}(t) = \pi_G(t)$.
	Otherwise, $H$ contains both the path $\pi_G^{-e}(v)$ and the path $\pi_G^{-e}(v,t) = \pi_G(v,t)$ so we can write:
	\[
		d_H^{-e}(t) \le d_G^{-e}(v) + d_G(v,t) \le d_G^{-e}(t) +2 d_G(v,t) \le d_G^{-e}(t) +2 d_G(t) \le 3 d_G^{-e}(t).
	\] \qed
\end{proof}

First, we give a high-level description of our algorithm for computing a $(1+\epsilon)$-\texttt{EASPT} (see Algorithm~\ref{alg:1_epsilon_ftspt}).
We build our structure $H$ by starting from the $3$-\texttt{EASPT} of size $O(n)$ returned by Algorithm~\ref{alg:3ftspt}. Then, our algorithm works in $n-1$ {\em phases}, where each phase considers the failure of an edge of $T$ w.r.t. a fixed preorder visit of the edges, say $e_1,\dots,e_{n-1}$. Let $e=e_h$ be the edge of $T$ of the $h$-th phase of the algorithm. The algorithm checks all the vertices of $G$ in preorder w.r.t.\ $T_G^{-e}(s)$. Whenever a vertex $t$ {\em is bad for} $e$, i.e., $d_{H}^{-e}(t)>(1+\epsilon)d^{-e}(t)$, the algorithm chooses a suitable value $\ell\geq 1$ and adds to $H$ all the last $\ell$ edges of $\pi^{-e}(t)$ that are missing. Notice that all the bad vertices for $e$ must necessarily belong to $D_G(e)$.
As we will see, the choice of $\ell$ is done so that we do not only guarantee that $d_H^{-e}(t)\leq (1+\epsilon)d^{-e}(t)$, but we also obtain a substantial improvement on the stretch factors of the distances from $s$ in $H-e$ for all the first $\ell-1$ predecessors of $t$ in $\pi^{-e}(t)$, say $t_1,\dots,t_{\ell-1}$. Furthermore, as we will prove later, the structure $H$ built by the algorithm after that edge $e$ has been considered guarantees the following  property: For every $h<i\leq n-1$, and for every $x \in \{t\}\cup\{t_1,\dots,t_{\ell-1}\}$, we have that $d_{H}^{-e_i}(x)\leq d_{H}^{-e}(x)$. These are exactly the two key ingredients for the analysis of our algorithm that, combined altogether, allow us to prove that each vertex causes the addition of $O(\epsilon^{-2}\log n)$ edges to $H$ on the average.

The main result we are going to prove in this section is the following:

\begin{theorem}
Given an $n$-vertex non-negatively real weighted graph $G$, a source vertex $s \in V(G)$, and any $\epsilon >0$, the structure $H$ returned in polynomial time by Algorithm~\ref{alg:1_epsilon_ftspt} is a $(1+\epsilon)$-\texttt{EASPT} of $G$ rooted at $s$ of size $O(\frac{n \log n}{\epsilon^2})$.
\end{theorem}

We will prove separately the bound on the stretch factor and on the size of the structure in the next two subsections (see Lemmas \ref{lemma:1_epsilon_ftspt_correctness} and \ref{lemma:size of H}, respectively).

\subsection{Stretch factor of the structure}
Observe that, for every bad vertex $t$ for $e$, the algorithm always adds the last edge of $\pi^{-e}(t)$ to $H$. This is enough to prove the correctness of our algorithm as the following lemma shows.

\begin{lemma}
 \label{lemma:1_epsilon_ftspt_correctness}
 The structure $H$ returned by Algorithm~\ref{alg:1_epsilon_ftspt} is a $(1+\epsilon)$-\texttt{EASPT}.
\end{lemma}
\begin{proof}
 Let $\tilde{H}$ be the structure built by the algorithm just before a bad vertex $t$ for an edge $e$ of $T$ is considered.
 Assume by induction that, for every vertex $x$ which has been visited by the algorithm before $t$ in the phase associated with $e$, we have $d^{-e}_{\tilde{H}}(x) \le (1+\epsilon) d^{-e}(x)$.
 Let $f=(t',t)$ be the last edge of $\pi^{-e}(t)$. By induction, $d^{-e}_{\tilde{H}}(t') \le (1+\epsilon) d^{-e}(t')$. Furthermore, as $t$ is bad for $e$, the algorithm adds $f$ to $\tilde{H}$. Since $\tilde{H}+f$ is a subgraph of $H$, we have that
 \begin{multline*}
  d_H^{-e}(t) \le d^{-e}_{\tilde{H} + f}(t) \le d^{-e}_{\tilde{H}}(t') + w(f) \le (1+\epsilon) d^{-e}(t') + w(f) \\
   \le (1+\epsilon) ( d^{-e}(t') + d^{-e}(t',t) ) = (1+\epsilon) d^{-e}(t).
 \end{multline*} \qed
\end{proof}

\section{A $(1+\epsilon)$-\texttt{EASPT} structure}
\label{sct:EASPT}
\begin{algorithm}[!t]
	\DontPrintSemicolon
	\SetKwInOut{Input}{Input}
	\SetKwInOut{Output}{Output}
	
	\Input{A non-negatively real weighted graph $G$, $s \in V(G)$, $\epsilon>0$}
	\Output{A $(1+\epsilon)$-\texttt{EASPT} of $G$ rooted at $s$}

	\BlankLine	
		
	$H \gets$ the 3-\texttt{EASPT} of size $O(n)$ returned by Algorithm~\ref{alg:3ftspt}.
%~\ref{alg:3ftspt} (see the Appendix).
	
	\For{$e \in E(T_G(s))$ in preorder w.r.t.\ $T_G(s)$}
	{
		\For{$t \in V(G)$ in preorder w.r.t.\ $T_G^{-e}(s)$}
		{
			\If(\tcc*[f]{vertex $t$ is \emph{bad} for edge $e$}){$d_H^{-e}(t) > (1+\epsilon)d_G^{-e}(t)$ }
			{
				Select a set of edges $S \subseteq E(\pi_G^{-e}(t))$ (see details after Lemma \ref{lemma:1_epsilon_ftspt_correctness})\;
				$H \gets H + S$ \;
			}
		}
	}
	\Return $H$
	
	\caption{Algorithm for building an $(1+\epsilon)$-\texttt{EASPT}}\label{alg:1_epsilon_ftspt}
\end{algorithm}

\subsection{Size of the structure}
Now we describe the edge selection process and we analyze the size of our final structure. Let $H_0$ be the initial $3$-\texttt{EASPT} structure. Let us fix the failed edge $e=(u,v)$ and a single bad vertex $t$ for $e$.
We call $H'$ the structure built by the algorithm just before $t$ is considered.
Let $f=(x,y)$ with $x \in U_G(e)$ be the unique edge in $C_{G}(e) \cap E(\pi_G^{-e}(t))$. Consider the subpath of $\pi_G^{-e}(t)$ going from $x$ to $t$ and let $x_0, x_1, \dots,x_r$ be its vertices, in order. We consider the set $Z=\{ x_i \, : \, (x_{i-1}, x_i) \not\in E( H^\prime ) \}$, we name its vertices $z_1, \dots, z_k = t$ with $k=|Z|$, in order and we let $z_0=x$ (see Figure~\ref{fig:tree_edge}).
We define $\alpha_i = \frac{d^{-e}_{H'}(z_i)}{d^{-e}(z_i)}$. It follows from the definitions and from the proof of Lemma~\ref{lemma:1_epsilon_ftspt_correctness} that we have $\alpha_0 = 1$, $\alpha_j \le (1+\epsilon)$ for $1 \le j < k$ and $\alpha_k > 1+\epsilon$.

\begin{figure}[t]
	\centering
	\includegraphics[scale=0.90]{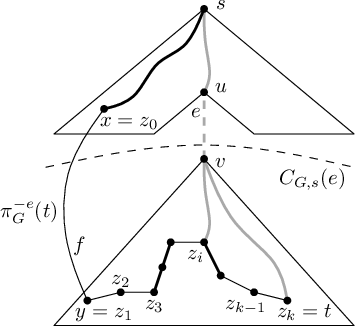}
	\caption{Edge selection phase of Algorithm~\ref{alg:1_epsilon_ftspt} when a bad vertex $t$ for the failing edge $e$ is considered. Bold edges belong to $H'$ while the black path is $\pi_G^{-e}(t)$.}
	\label{fig:tree_edge}
\end{figure}

Think of the edges in $\pi^{-e}(t)$ as being directed towards $t$ for a moment.
In the following we will describe how to select the set $S$ of edges used by the algorithm. In particular, we will select $\eta \ge 1$ edges entering into the last $\eta$ vertices in $Z$.
This choice of $S$ will ensure that the overall decrease of the values $\alpha_i$ in $H' + S$ will be at least $\frac{\epsilon}{ \hn_n} \eta$ where $\hn_n$ denotes the $n$-th \emph{harmonic number}. In particular, we exploit the fact that, after adding the set $S$, each ``new value'' $\alpha_i$ with $i > k - \eta$, will not be larger than $\alpha_{k-\eta}$ as we will show in the following.

%Let us define $U_h^t$ as the set of vertices for which an incoming edge is added to the structure while considering the failure of the edge $e_h$ and the vertex $t$.
%Moreover let $U_h = \bigcup_{t \in T_h} U_h^t$.
%In a similar way we define $V_h$ to be the set of vertices for which an incoming edge is added to the structure \emph{before or while} the failure of edge $e_h$ is considered, i.e., $V_h = \bigcup_{i=1}^h U_i$.

Consider the sequence $\gamma_0, \dots, \gamma_k$ where $\gamma_i = 1 + \frac{\epsilon}{\hn_k}  (\hn_k - \hn_{k-i})$. Notice that the sequence is monotonically increasing from $\gamma_0=1$ to $\gamma_k = 1+\epsilon$.
Let $0 \le j < k$ be the largest index such that $\alpha_j \le \gamma_j$. Notice that $j$ always exists as $\alpha_0 = \gamma_0$ and $\alpha_k > \gamma_k$. We set $\eta=k-j$ so that the set $S$ is defined accordingly. Let $U=\{z_{j+1}, \dots, z_k\}$ be the set of vertices for which an incoming edges has been added in $S$.
%Let $m_h$ be the overall number of edges added by the algorithm when considering the failure of the edge $e_h$.

For every vertex $z \in U$ we define the following path in $H' + S- e$: $P(z) = \pi_{H'}^{-e}(z_j) \circ \pi(z_j, z)$. Notice that $\pi(z_j, z)$ is entirely contained in $H' + S - e$.
We define $\alpha'_i = \frac{w(P(z_i))}{d^{-e}(z_i)}$, and note that $\alpha'_i$ is an upper bound to the stretch of $z$ in $H' + S - e$.

\begin{lemma}
\label{lemma:alpha_order}
For $i > j$, $\alpha'_i \le \alpha_j < \alpha_i$.
\end{lemma}
\begin{proof}
By definition of $j$, we have $\alpha_j \le \gamma_j < \gamma_i < \alpha_i$.
Now we prove $\alpha'_i \le \alpha_j$:
\[
	\alpha'_i = \frac{w(P(z))}{d^{-e}(z_i)} \!=\! \frac{d_{H'}^{-e}(z_j) + d(z_j, z_i)}{d^{-e}(z_i)} \!\le\!
	\frac{\alpha_j d^{-e}(z_j) + d^{-e}(z_j, z_i)}{d^{-e}(z_i)} \!\le\! \frac{\alpha_j d^{-e}(z_i)}{d^{-e}(z_i)} = \alpha_j.
\]\qed
\end{proof}

\begin{figure}
	\centering
	\includegraphics[scale=1]{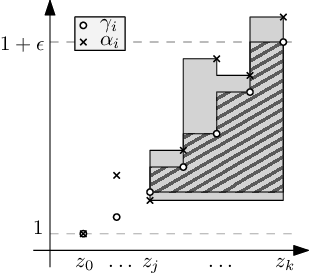}
	\caption{Representation of the sequences $\alpha_i$ and $\gamma_i$. The gray area is a lower bound to the overall decrease of the $\alpha_i$ values w.r.t.\ $\alpha'_i$ with $i > j$. This area is, in turn, lower bounded by the area of the striped region which is $\frac{\epsilon}{\hn_k}(k-j)$.}
	\label{fig:stretch_decrease}
\end{figure}

We now lower-bound the overall decrease of the values $\alpha'_i$'s w.r.t.\ the corresponding $\alpha_i$'s by using the following inequalities (see Figure~\ref{fig:stretch_decrease}):
\begin{multline}
	\label{eq:stretch_decrease_single_bad_vertex}
	\sum_{z \in U} \left( \frac{d^{-e}_{H'}(z)}{d^{-e}(z)} - \frac{w(P(z))}{d^{-e}(z)} \right) =
	\sum_{i=j+1}^k (\alpha_i - \alpha'_i) \ge \sum_{i=j+1}^k (\alpha_i - \alpha_j) \ge \sum_{i=j+1}^k (\gamma_i - \gamma_j) \\
	= \frac{\epsilon}{\hn_k} \sum_{i=j+1}^k \left( \hn_{k-j} - \hn_{k-i} \right)
	= \frac{\epsilon}{\hn_k} (k-j)
	\ge \frac{\epsilon}{\hn_n} \eta.
\end{multline}

\noindent
where in the last but one step we used the well-known equality that for every $j \le k$, $\sum_{i=j+1}^k \left( \hn_{k-j} - \hn_{k-i} \right) = k-j$.

The above selection procedure is repeated by the algorithm for every failed edge $e_h$ and for every corresponding bad vertex. We now focus on the $h$-th phase of the algorithm.
We call $B_h$ the set of all the bad vertices considered in this phase and, for every $t \in B_h$, we call $U(t)$ the corresponding set $U$, as defined above.
Moreover, let $U_h = \cup_{t \in B} U(t)$, and let $V_h = \bigcup_{i=1}^h U_i$ (notice that $V_0 = \emptyset$).
%For a vertex $z \in U_h$, let $B_z = \{ t \in B \, : \,  z \in U(t) \}$ and, given $t \in B_z$, call $P^t(z)$ the path $P(z)$ built when considering the bad vertex $t$, as defined above.
Notice that the sets $U(t)$ are pairwise disjoint since, once $z \in U(t)$ we add the edge of $\pi^{-e}(t)$ entering $z$ and hence $z$ cannot belong to any other set $U(t')$ where $t'$ is a bad vertex which is considered after $t$ in phase $h$. Hence, we let $P_h(z)$ be the unique path $P(z)$ which is built during phase $h$.
Finally, let  $H'_h$, $H(t)$, and $H_h$ be the structures built by the algorithm at the beginning of phase $h$, just before the bad vertex $t \in B_h$ is processed, and at the end of phase of phase $h$, respectively.

Let $m_h$ be the number of new edges added during the phase $h$. We can now prove:
\begin{lemma}
	\label{lemma:stretch_decrease}
	$\displaystyle \sum_{z \in U_h} \left( \frac{d^{-e_h}_{H'_h}(z)}{d^{-e_h}(z)} - \frac{w(P_h(z))}{d^{-e_h}(z)} \right) \ge m_h \frac{\epsilon}{\hn_n}$.
\end{lemma}
\begin{proof}
	For a bad vertex $t \in B_h$, let $\eta_t$ be the number of edges selected by the algorithm, i.e. $|S|$, when $t$ is considered.
	By summing Equation~\ref{eq:stretch_decrease_single_bad_vertex} over all vertices $t \in B_h$, we obtain:
	\begin{multline*}
		\sum_{z \in U_h} \left( \frac{d^{-e_h}_{H'_h}(z)}{d^{-e_h}(z)} - \frac{w(P_h(z))}{d^{-e_h}(z)} \right)
		= \sum_{t \in B_h} \sum_{z \in U(t)} \left( \frac{d^{-e_h}_{H'_h}(z)}{d^{-e_h}(z)} - \frac{w(P_h(z))}{d^{-e_h}(z)} \right) \\
		\ge \sum_{t \in B_h} \sum_{z \in U(t)} \left( \frac{d^{-e_h}_{H(t)}(z)}{d^{-e_h}(z)} - \frac{w(P_h(z))}{d^{-e_h}(z)} \right)
		\ge \sum_{t \in B_h}  \eta_t \frac{\epsilon}{\hn_n} \ge m_h \frac{\epsilon}{\hn_n},
	\end{multline*}
	where we used the facts that the sets $U(t)$ are pairwise disjoint, and that every $H(t)$ is a supergraph of $H'_h$.\qed
\end{proof}

Now, let us define a function $\phi_h(z)$, first we set $\phi_0(z) = \frac{6}{\epsilon}d(z)$ for every $z \in V$, and then we recursively define:
\[
	\phi_h(z) = \begin{cases}
		w(P_h(z)) & \mbox{if } z \in U_h \\
		\phi_{h-1}(z) & \mbox{if } z \not\in U_h
	\end{cases}
\]

\noindent We will show that if $z \in V_h$, then $\phi_{h-1}(z)$ is an upper bound to $d_{H'_h}^{-e_h}(z)$. In order to do so we separately consider the cases $z \in U_h \setminus V_{h-1}$ and $z \in V_{h-1}$ in the following two lemmas:

\begin{lemma}
	\label{lemma:stretch_bound}
	For every $z \in U_h \setminus V_{h-1}$ we have $\phi_{h-1}(z) \ge d_{H'_h}^{-e_h}(z)$.
\end{lemma}
\begin{proof}
	Since $z \in U_h$ we know that an incoming edge to $z$ has been selected when the algorithm was considering some bad vertex $t$ for the edge $e_h$. We have:
	\begin{multline*}
		d^{-e_h}(v) + d(v,t) > (1+\epsilon) d^{-e_h}(t) = (1+\epsilon)( d^{-e_h}(z) + d(z, t) ) \\
		\ge  (1+\epsilon)( d^{-e_h}(z) + \left| d(z) - d(t) \right| ).
	\end{multline*}
	\noindent Moreover, we also have:
	\[
		d^{-e_h}(v) + d(v,t) \le  d^{-e_h}(z) + d(z, v) + d(v,t) \le d^{-e_h}(z) + d(z) + d(t)
	\]
	
	\noindent The above inequalities together imply:
	\[
		d^{-e_h}(z) < \frac{d(z)+d(t)- (1+\epsilon) \left| d(z) - d(t) \right|}{\epsilon}.
	\]

	\noindent If $d(z) \ge  d(t)$, the above formula becomes:
	\[
		d^{-e_h}(z) < \frac{d(z)+d(t)- (1+\epsilon) ( d(z) - d(t) ) }{\epsilon} = \frac{(2+\epsilon) d(t) - \epsilon d(z) }{\epsilon} \le \frac{ 2 d(z) }{\epsilon}.
	\]
	\noindent Otherwise, $d(z) < d(t)$ and we have:
	\[
		d^{-e_h}(z) < \frac{d(z)+d(t)- (1+\epsilon) (  d(t) - d(z)  ) }{\epsilon} = \frac{(2+\epsilon) d(z) - \epsilon d(t) }{\epsilon} \le \frac{ 2 d(z) }{\epsilon}.
	\]
	
	\noindent As $H_h$ is a supergraph of $H_0$, which is a $3$-\texttt{EASPT}, we immediately have:
	\[
		d_{H_h'}^{-e_h}(z) \le d_{H_0}^{-e_h}(z)\le 3 d^{-e_h}(z) < \frac{6}{\epsilon} d(z) = \phi_{h-1}(z).
	\]\qed
\end{proof}

We now consider the remaining case:
\begin{lemma}
	\label{lemma:phi_ub}
	For $z \in V_{h-1}$, $\phi_{h-1}(z) \ge d_{H'_h}^{-e_h}(z)$.
\end{lemma}
\begin{proof}
	We show that the weight of every path $P_h(z)$ built by the algorithm when $e_h=(u,v)$ fails is an upper bound to  $d^{-e_{h'}}_{H'_{h'}}(z)$ for every $h' > h$. This will immediately imply the claim.
	To prove the above we argue that $P_h(z)$ is vertex disjoint from $\pi^{-e_h}(v,z)$ (except for $z$).	
	As a consequence, when $e_{h'}$ fails either $P_h(z)$ is still in $H'_{h'} - e_{h'}$ or $e_{h'}$ is not in $\pi(z)$, hence $d^{-e_{h'}}_{H'_{h'}}(z) = d(z)$.
	
	Let $H'$ be the structure constructed by the algorithm just before $P_h(z)$ is built, $t$ be the corresponding bad vertex, and $z_j$ be the vertex chosen as described above.
	Recall that $z=z_i$ for some $i > j$, and that $P_h(z) = \pi_{H'}^{-e_h}(z_j) \circ \pi(z_j, z)$.
	Suppose by contradiction that $P_h(z)$ and $\pi(v, z)$ intersect at some vertex $q \not= z$. Clearly $q \in D(e)$.
	If $q \in V(\pi(z_j, z))$ then $\pi(z_j, z)$ contains $\pi(q, z)$ as a subpath.\footnote{This is due to the tie-breaking rule discussed before which gives priority to edges in $T$.} As $\pi(z_j, z)$ is, in turn, a subpath of $\pi^{-e_h}(z)$, this implies that the edge preceding $z$ in $\pi^{-e_h}(z)$ belongs to $T \subseteq H_0$ and this contradicts the definition of $z$.
	
	Otherwise $q \in V(\pi_{H'}^{-e_{h}}(z_j))$. As $z_j$ precedes $z$ in $\pi^{-e}(t)$ we have $d^{-e}(z_j) \le d^{-e}(z)$. Since $\pi^{-e}(q,z) = \pi(q,z)$ which is in $H'$, we can write:
	\begin{multline*}
		\alpha_i = \frac{ d_{H'}^{-e}(z) }{d^{-e}(z)} \le \frac{ d_{H'}^{-e}(q) + d^{-e}(q, z) }{d^{-e}(z)} \le
		\frac{ d_{H'}^{-e}(q) + d^{-e}(q, z_j) + d^{-e}(z_j, z) }{d^{-e}(z)} = \\
		 \frac{ d_{H'}^{-e}(z_j) + d^{-e}(z_j, z) }{d^{-e}(z_j) + d^{-e}(z_j, z)} \le \max\left\{ \frac{ d^{-e}_{H'}(z_j)}{d^{-e}(z_j)}, \frac{d^{-e}(z_j, z)}{d^{-e}(z_j, z)}\right\} = \max\{ \alpha_j, 1 \}  = \alpha_j
	\end{multline*}
	where we used that for every $a,b,c,d>0$, we have that $\frac{a+b}{c+d} \le \max\left\{ \frac{a}{c}, \frac{b}{d}  \right\}$, and the inequality $\alpha_j \ge 1$.
	The above contradicts Lemma~\ref{lemma:alpha_order}.\qed
\end{proof}

To summarize, combining Lemma~\ref{lemma:stretch_bound} and \ref{lemma:phi_ub} together, we immediately have:
\begin{corollary}
	\label{corollary:phi_ub}
	If $z \in V_h$, then $\phi_{h-1}(z) \ge d_{H'_h}^{-e_h}(z)$.
\end{corollary}

Next lemma shows that $\phi_h(z)$ is monotonically non-increasing w.r.t.\ $h$:

%\begin{lemma}
%	\label{lemma:path_improves}
%	For $z \in U_h$, $d_{H'_h}^{-e_h}(z) \ge w(P_h(z))$.
%\end{lemma}

\begin{lemma}
	\label{lemma:path_improves}
	For every $h \ge 1$, $\phi_{h-1}(z) \ge \phi_{h}(z)$.
\end{lemma}
\begin{proof}
	If $z \not\in U_h$ then, by definition, we have $\phi_{h-1}(z) = \phi_{h}(z)$.
	Otherwise $z \in U_h$, let $\tilde{H}$ be the structure constructed by the algorithm just before $P_h(z)$ is built and recall that $z=z_i$ for some $i$. As $z \in U_h$ we have $H'_h \subseteq \tilde{H}$, moreover Corollary~\ref{corollary:phi_ub} holds, hence we can write $\phi_{h-1}(z) \ge d_{H'_h}^{-e_h}(z) \ge d_{\tilde{H}}^{-e_h}(z)$. The claim follows as we have:
	\[
		d_{\tilde{H}}^{-e_h}(z) \ge \phi_{h}(z) \iff \frac{d_{\tilde{H}}^{-e_h}(z)}{d^{-e_h}(z)} \ge \frac{w(P_h(z))}{d^{-e_h}(z)} \iff \alpha_i \ge \alpha'_i
	\]
	\noindent which is true by Lemma~\ref{lemma:alpha_order}.\qed
\end{proof}

We now define a non-increasing global potential function $\Phi(h)$ for $0<h \le n-1$:
\[
	\Phi(h) = \sum_{z \in V \setminus \{s\}} \frac{\phi_h(z)}{d(z)}.
\]

\noindent
The following lemma bounds the decrease of $\Phi$ after each phase of the algorithm:
\begin{lemma}
	\label{lemma:stretch_delta}
	$\Phi(h-1) - \Phi(h) \ge m_h \frac{\epsilon}{\hn_n}$.
%	$\displaystyle \sum_{z \in U_h}  \frac{\phi_{h-1}(z)}{d(z)}  - \sum_{z \in U_h}  \frac{\phi_h(z)}{d(z)}   \ge m_h \frac{\epsilon}{\hn_n} - |U_h \setminus V_{h-1}| \frac{6}{\epsilon}$.
\end{lemma}
\begin{proof}
	Using the definitions we have:
	\begin{multline*}
		\Phi(h-1) - \Phi(h) = \sum_{z \in V \setminus \{s\}} \frac{\phi_{h-1}(z)}{d(z)} - \sum_{z \in V \setminus \{s\}} \frac{\phi_{h}(z)}{d(z)} = \\
		\sum_{z \in U_h} \frac{\phi_{h-1}(z) - \phi_{h}(z)}{d(z)} + \sum_{z \in V \setminus (U_h \cup \{s\})} \!\! \frac{\phi_{h-1}(z) - \phi_{h}(z)}{d(z)} = \sum_{z \in U_h} \frac{\phi_{h-1}(z) - \phi_{h}(z)}{d(z)}		
	\end{multline*}

	\noindent where the latter equality follows from the fact that $\phi_h(z) = \phi_{h-1}(z)$ whenever $z \not\in U_h \cup \{s\}$. Starting from the latter quantity, we use Lemma \ref{lemma:path_improves} and then Corollary~\ref{corollary:phi_ub} to write:
	\[
		\sum_{z \in U_h} \frac{\phi_{h-1}(z) - \phi_{h}(z)}{d(z)} \ge \sum_{z \in U_h} \frac{\phi_{h-1}(z) - \phi_{h}(z)}{d^{-e_h}(z)} \ge \sum_{z \in U_h} \left( \frac{ d_{H'_h}(z)}{d^{-e_h}(z)} - \frac{w(P_h(z))}{d^{-e_h}(z)} \right).
	\]
	\noindent which is at least $m_h \frac{\epsilon}{\hn_n}$ as shown by Lemma~\ref{lemma:stretch_decrease}.\qed
\end{proof}

	We are finally able to prove the following:
\begin{lemma}\label{lemma:size of H}
	The size of the structure $H$ returned by Algorithm~\ref{alg:1_epsilon_ftspt} is $O(\frac{n \log n}{\epsilon^2})$.
\end{lemma}
\begin{proof}
	Since $H_0$ contains $O(n)$ edges, we only focus on bounding the number $\mu = \sum_{h=1}^{n-1} m_h$ of edges in $E(H)\setminus E(H_0)$.
	Notice that, by definition of $\phi_0(z)$, we have $\Phi(0) = \sum_{z \in V \setminus \{s\}} \frac{\phi_0(z)}{d(z)} \le \frac{6}{\epsilon} n$. Moreover, as every $\phi_h(z)$ is non-negative, $\Phi(n-1) \ge 0$ holds.
	Using these inequalities together with Lemma~\ref{lemma:stretch_delta}, we can write:	
	\[
		\frac{6}{\epsilon} n \ge \Phi(0) - \Phi(n-1) = \sum_{h=1}^{n-1} \left( \Phi(h-1) - \Phi(h) \right) \ge
	\frac{\epsilon}{\hn_n} \sum_{h=1}^{n-1} m_h = \frac{\epsilon}{\hn_n} \mu
	\]
	which can be solved for $\mu$ to get $\mu= O(\frac{n \log n}{\epsilon^2})$.\qed
\end{proof}

\section{A $(1+\epsilon)$-\texttt{VASPT} structure}\label{section:VASPT}
\label{sct:VASPT}
In this section we extend our previous $(1+\epsilon)$-\texttt{EASPT} structure to deal with vertex failures.
In order to do so we build a different initial subgraph $H_0$, which is a $3$-\texttt{VASPT} having suitable properties that we will describe later.
Then we use the natural extension of Algorithm~\ref{alg:1_epsilon_ftspt} where we consider (in preorder) vertex failures instead of edge failures.

The construction of the subgraph $H_0$ is similar to that given by Baswana and Khanna \cite{BK13} for the related problem of computing a vertex-fault-tolerant {\ttfamily SDSO} which reports (post-failure) $3$-approximate distances from $s$. In particular, the key difference between their  construction and ours is pointed out within the proof of the forthcoming Lemma \ref{lemma:bad_vertices_down}, and such a difference is instrumental to guarantee the correctness of our approach.
%The difference between our construction and the one in \cite{BK13} is that, for every failing vertex $u$ and for every $t$
In the following, we first describe the construction of our structure $H_0$, and then we argue on how the analysis for the edge-failure case can be adjusted to show the same bound on the size of $H$ for the vertex failure case as well.

%The structure $H_0$ is initially equal to $T$ and it is augmented by using a technique similar to the one shown in \cite{BK13}: the SPT
Initially, $H_0$ is equal to $T$. Then, proceeding as proposed in \cite{BK13}, $T$ is decomposed into ancestor-leaf vertex-disjoint paths in the following recursive way: select a path $Q$ from the root of $T$ to a leaf such that the removal of $Q$ splits the tree into a forest where the size of each subtree is at most half the size of the original tree, and then proceed recursively on each subtree.
After this preliminary path-decomposition step of $T$, for each generated path an approximate structure is built. This structure will provide approximate distances towards the vertices $V(G) \setminus \{ u \}$ whenever any vertex $u$ along the path fails. The union of $T$ with all these structures will form $H_0$.

Let us then describe how to build the initial structure for a fixed path $Q$ of the previous decomposition. Let $q$ be the starting vertex of $Q$, and let $T_q$ be the subtree of $T$ rooted at $q$. Moreover, let $u \in V(Q)$ be a failing vertex, and let $v$ be the next vertex in $Q$.\footnote{W.l.o.g. we are assuming that the failing vertex $u$ is not a leaf, as otherwise $T - u$ is already a SPT of $G - u$.} Similarly to what is done in \cite{BK13}, we partition the vertices of the forest $T - u$ into three sets: (i) the \emph{up set} $U$ containing all the vertices of the tree rooted at $s$, (ii) the \emph{down set} $D$ containing all the vertices of the tree rooted at $v$, and (iii) the \emph{others set} $O$ containing all the remaining vertices (see Figure~\ref{fig:tree_vertex}).

\begin{figure}[!t]
	\centering
	\includegraphics[scale=.85]{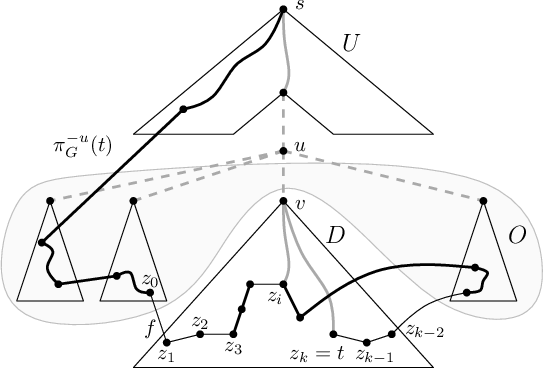}
	\caption{Edge selection phase of the vertex-version of Algorithm~\ref{alg:1_epsilon_ftspt} when a bad vertex $t$ for the failing vertex $u$ is considered. Bold edges belong to $H'$ while the black path is $\pi_G^{-u}(t)$. Notice that all $z_i$s belong to the down set $D$.}
	\label{fig:tree_vertex}
\end{figure}

In order to select the set of additional edges associated with $Q$, we construct a SPT $T'$ of $G - u$ and we imagine that its edges are directed towards the leaves. We select all the edges of $E(T') \setminus E(T)$ that do not lead to a vertex in $D$, plus the unique edge of $\pi^{-u}(v)$ that crosses the cut induced by the sets $U \cup O$ and $D$.
Notice that $T - u$ contains all the paths in $T'$ towards the vertices in $U$, and that each vertex has at most one incoming edge in $T'$. This implies that the number of selected edges is at most $|O|+1$.

The above procedure is repeated for all the failing vertices of $Q$, in order. As the sets $O$ associated with the different vertices are disjoint we have that, while processing $Q$, at most $|V(T_q)|+|Q| = O(|V(T_q)|)$ edges are selected. Finally, the procedure is repeated for all the paths of the decomposition, and since such a decomposition is done as suggested in \cite{BK13}, it immediately follows that the size of the entire structure $H_0$ is $O(n \log n)$.

We now prove some useful properties of the structure $H_0$. First of all, observe that, by construction and similarly to the edge-failure case, we immediately have:
\begin{lemma}
	\label{lemma:vertex_stretch}
	Let $u$ be a failed vertex and consider another vertex $z \not= u$. We have: (i) $d_{H_0}^{-u}(v) = d^{-u}(v)$, and (ii) for $z \in D$, it holds $d_{H_0}^{-u}(z) \le 3 d^{-u}(z)$.
\end{lemma}

Moreover, we also have the following:

\begin{lemma}
	\label{lemma:bad_vertices_down}
	Consider a failed vertex $u$. During the execution of the vertex-version of Algorithm~\ref{alg:1_epsilon_ftspt}, every bad vertex $t$ for $u$ will be in $D$.
\end{lemma}
\begin{proof}
	Let $\tilde{H}$ be the structure build by the algorithm just before $t$ is considered.
	Assume by contradiction that $t \not\in D$. Clearly, $t$ cannot be in $U$ so we must have $t \in O$.
	By construction of $H_0$, the path $\pi^{-u}(t)$ must contain some vertex of $D$. Let $z$ be the last vertex of $\pi^{-u}(t)$ that is also in $D$. As $z$ precedes $t$ in $\pi^{-u}(t)$ we must have $d^{-u}_{\tilde{H}}(z) \le (1+\epsilon) d^{-u}(z)$. Moreover, by construction, $\pi^{-u}(z,t)$ is entirely contained in $H_0$.\footnote{Notice that this property would not be guaranteed by the initial structure provided in \cite{BK13}, and it is exactly the key difference between our construction and the one given in \cite{BK13}.} This implies:
	\[
		d_{\tilde{H}}^{-u}(t) \le d_{\tilde{H}}^{-u}(z) +  d^{-u}(z, t) \le (1+\epsilon) d^{-u}(z) + d^{-u}(z,t) \le  (1+\epsilon) d^{-u}(t).
	\]
	\noindent which contradicts the fact that $t$ is a bad vertex for $u$.\qed
\end{proof}

At this point, the same analysis given for the case of edge failures can be retraced for vertex failures as well. We point out that  Lemma~\ref{lemma:bad_vertices_down} ensures that every bad vertex for $u$ is in the same subtree as $v$. Also notice that all the vertices $z_i$'s are, by definition, in the same subtree as well (see Figure~\ref{fig:tree_vertex}). The above, combined with Lemma~\ref{lemma:vertex_stretch} (i), is needed by the proof of Lemma~\ref{lemma:stretch_bound}, while Lemma~\ref{lemma:vertex_stretch} (ii) is used in the proof of Lemma~\ref{lemma:stretch_delta}.
Hence we have:

\begin{theorem}
	Given an $n$-vertex non-negatively real weighted graph $G$, a source vertex $s \in V(G)$, and any $\epsilon >0$, the vertex-version of Algorithm~\ref{alg:1_epsilon_ftspt} computes in polynomial time a $(1+\epsilon)$-\texttt{VASPT} of $G$ rooted at $s$ of size $O(\frac{n \log n}{\epsilon^2})$.
\end{theorem}

\section{Relation with $(\alpha, \beta)$-spanners in unweighted graphs}
\label{sct:unweighted}
In this section we turn our attention to the unweighted case, and we provide two polynomial-time algorithms that augment an $(\alpha,\beta)$-spanner of $G$ so to obtain an $(\alpha,\beta)$-\texttt{EABFS}/\texttt{VABFS}. We present the algorithm for the vertex-failure case, and then we show how it can be adapted to the edge-failure case.

The algorithm first augments the structure $H_0$ computed as explained in Section~\ref{section:VASPT}, and then adds its edges to the $(\alpha,\beta)$-spanner of $G$. The structure $H_0$ is augmented as follows. The vertices of the BFS of $G$ rooted at $s$ are visited in preorder. Let $u$ be the vertex visited by the algorithm and let $D$ be the set of vertices of the tree defined so as explained in Section~\ref{section:VASPT} w.r.t the path decomposition computed for $H_0$. For every $t \in D$, the algorithm checks whether $\pi_G^{-u}(s,t)$ contains no vertex of $D\setminus\{t\}$ and $d_G^{-u}(s,t)<d_{H_0}^{-u}(s,t)$. If this is the case, then the algorithm augments $H_0$ with the edge of $\pi_G^{-u}(s,t)$ incident to $t$.

The following observation is crucial to prove the algorithm correctness.

\begin{fact}\label{fact:node_failure_bfs}
For every vertex $u$ and every vertex $t \in V(G) \setminus \{u\}$ such that $\pi_G^{-u}(t)$ contains a vertex in $D$, let $x$ and $y$ be the first and last vertex of $\pi_G^{-u}(t)$ that belong to $D$, respectively. We have $d_{H_0}^{-u}(x)=d_{G}^{-u}(x)$ and
$d_{H_0}^{-u}(y,t)=d_{G}^{-u}(y,t)$.
\end{fact}

We can now give the following:

\begin{theorem}\label{th:from_spanner_towards_vabfs}
Given an unweighted graph $G$ with $n$ vertices and $m$ edges, a source vertex $s \in V(G)$, and an $(\alpha, \beta)$-spanner for $G$ of size $\sigma=\sigma(n,m)$, it can be computed in polynomial time an $(\alpha, \beta)$-\texttt{VABFS} w.r.t.\ $s$ of size $O\big(\sigma+n\log n\big)$.
\end{theorem}
\begin{proof}
Let $H$ be the subgraph of $G$ computed by the algorithm. We first prove that $H$ is an $(\alpha,\beta)$-\texttt{VABFS} of $G$ and $s$ by showing that $d_{H}^{-u}(s,t)\leq \alpha \cdot d_{G}^{-u}(s,t)+\beta$, for two distinct vertices $u,t \in V(G)$. W.l.o.g., we can assume that $\pi_G^{-u}(s,t)$ contains some vertices of $D\setminus\{t\}$ because, if our assumption was not true, then,  by Fact~\ref{fact:node_failure_bfs}, $d_{H}^{-u}(s,t)=d_{G}^{-u}(s,t)\leq \alpha \cdot d_G^{-u}(s,t)+\beta$.

Let $x$ and $y$ be the first and last vertex of $\pi_G^{-u}(s,t)$ contained in $D\setminus\{t\}$ in a path traversal from $s$ to $t$, respectively. We have that $\pi_G^{-u}(s,t)=\pi_G^{-u}(s,x)\circ \pi_G^{-u}(x,y)\circ \pi_G^{-u}(y,t)$, i.e.,
\begin{equation}\label{eq:shortest_path_decomposition}
d_G^{-u}(s,t)=d_G^{-u}(s,x)+d_G^{-u}(x,y)+d_G^{-u}(y,t).
\end{equation}
By Fact~\ref{fact:node_failure_bfs}, $H$ contains $\pi_G^{-u}(s,x)$ as well as $\pi_G^{-u}(y,t)$. Therefore,
\begin{equation}\label{eq:shortest_subpaths}
d_H^{-u}(s,x)=d_G^{-u}(s,x) \,\,\,\,\,\,\, \text{ and } \,\,\,\,\,\,\, d_H^{-u}(y,t)=d_G^{-u}(y,t).
\end{equation}
We now prove that $d_H^{-u}(x,y)\leq \alpha \cdot d_G^{-u}(x,y)+\beta.$
Since $H$ contains an $(\alpha,\beta)$-spanner of $G$, $H$ contains a path $P$ from $x$ to $y$ such that $w(P)\leq \alpha \cdot d_G(x,y)+\beta$. Clearly, if $u \not\in V(P)$, then $H-u$ contains $P$ and therefore $d_H^{-u}(x,y)\leq w(P)\leq \alpha \cdot d_G(x,y)+\beta$. Otherwise, if $u \in V(P)$, then let $v$ be the least common ancestor of $x$ and $y$ in the BFS of $G$ rooted at $s$. Since $v \in D$, it follows that
\begin{equation*}
d_H^{-u}(x,y)\leq d_H(x,v)+d_H(v,y)<d_H(x,u)+d_H(u,y)\leq w(P)\leq \alpha \cdot d_G(x,y)+\beta.
\end{equation*}
Using the last inequality together with Equations~(\ref{eq:shortest_path_decomposition}) and~(\ref{eq:shortest_subpaths}), we have that
\begin{align*}
d_H^{-u}(s,t) &\leq d_H^{-u}(s,x)+d_H^{-u}(x,y)+d_H^{-u}(y,t)\\
			  &\leq d_G^{-u}(s,x)+\alpha\cdot d_G(x,y)+\beta + d_G^{-u}(y,t)\leq \alpha\cdot d_G^{-u}(s,t)+\beta.
\end{align*}

We now prove that the size of $H$ is $O\big(\sigma+n\log n\big)$ by showing that the size of $H_0$ is $O(n\log n)$. We have already shown in the previous section that the number of edges of $H_0$ before the algorithm augments it is $O(n\log n)$. Therefore, it remains to bound the number of edges added to $H_0$. Let $F$ be the set of such edges. We prove that $|F|\leq 3n$ by showing that each vertex $t$ caused the addition of at most 3 edges to $F$. Let $t$ be a fixed vertex. Let $u_0,\dots,u_\ell$ be the vertices of the path $\pi_G(s,t)$, in a traversal of the path from $s$ to $t$ whose failures caused the insertion of the edge $(v_i,t)$ of $\pi_G^{u_i}(s,t)$ incident to $t$ in $F$. Since $G$ is unweighted, $d_G(s,t)=d_G(s,v_i)+j$, where $j\in\{-1,0,1\}$.
Furthermore, for every vertex $u'\neq t$ which is a proper descendent of $u_0$ in the BFS tree of $G$ rooted at $s$, $H-u'$ contains the path $\pi_G(s,v_0)\circ \pi_G(v_0,t)$ of length at most $d_G(s,t)+1+1=d_G(s,t)+2$.
Finally, observe that for every $1\leq i\leq \ell$ and for every vertex $u'\neq t$ which is a descendent of $u_i$ in the BFS tree of $G$ rooted at $s$, $H-u'$ contains the path $\pi_G^{-u_i}(s,t)$. Therefore, for every $2\leq i\leq \ell$, we have that
\[
d_G(s,t)\leq d_G^{-u_i}(s,t)\leq d_G(s,t)+2-i.
\]
The above inequality implies that $\ell\leq 2$. Hence each vertex $t$ caused the addition of at most $\ell+1\leq 3$ edges to $F$.\qed
\end{proof}

Now, we adapt the algorithm to prove a similar result for the $(\alpha,\beta)$-\texttt{EABFS}. The algorithm first augments a BFS tree $T$ of $G$ rooted at $s$ and then adds its edges to the $(\alpha,\beta)$-spanner of $G$. The tree $T$ is augmented by visiting its edges in preorder. Let $e$ be the edge visited by the algorithm. For every $t \in D_{G}(e)$, the algorithm checks whether $\pi_G^{-e}(s,t)$ contains no vertex of $D_{G}(e)\setminus\{t\}$ and $d_G^{-e}(s,t)<d_{T}^{-e}(s,t)$. If this is the case, then the algorithm augments $T$ with the edge of $\pi_G^{-e}(s,t)$ incident to $t$. %In the full version of the paper it will be shown that the proof of Theorem~\ref{th:from_spanner_towards_vabfs} can be adapted to prove the following:
It is easy to see that the proof of Theorem~\ref{th:from_spanner_towards_vabfs} can be adapted to prove the following: %TODO

\begin{theorem}
Given an unweighted graph $G$ with $n$ vertices, a source vertex $s \in V(G)$, and an $(\alpha, \beta)$-spanner for $G$ of size $\sigma$, it can be computed in polynomial time an $(\alpha, \beta)$-\texttt{EABFS} w.r.t.\ $s$ of size less than or equal to $\sigma+3n$.
\end{theorem}

Notice that the obtained $(\alpha, \beta)$-\texttt{E/VABFS} structures can be easily adapted to the multisource case, by simply rooting at each given source vertex $s \in S$ an augmented BFS. This will immediately provide corresponding $(\alpha, \beta)$-stretched sourcewise edge/vertex-fault-tolerant spanners (\texttt{SES/SVS}) of size $O\big(\sigma+|S| \cdot n\big)$ and $O\big(\sigma+|S| \cdot n\log n\big)$, respectively.

Interestingly, this immediately allows to improve some existing contructions. Indeed, by using the $(1,4)$-spanner of size $\widetilde{O}(n^\frac{7}{5})$ given in \cite{Che13} we obtain the following result:

\begin{corollary}
Given an unweighted graph $G$ with $n$ vertices, and a set of source vertices $S \subseteq V(G)$,
we can compute in polynomial time a $(1,4)$-\texttt{SES} of $G$ w.r.t.\ $S$ having size $\widetilde{O}(n^{\frac{7}{5}}+ |S| \cdot n )$.
\end{corollary}

\noindent
This sparsifies the $(1,4)$-\texttt{SES} of size $O(|S| \cdot n^{\frac{4}{3}})$ given in \cite{PP14} as soon as $|S| = \widetilde{\omega}(n^\frac{1}{15})$.

Moreover, by using the $(1,6)$-spanner of size $O(n^{4/3})$ provided in \cite{BKMP10}, we also have:

\begin{corollary}
Given an unweighted graph $G$ with $n$ vertices, and a set of source vertices $S \subseteq V(G)$,
we can compute in polynomial time a $(1,6)$-\texttt{SVS} of $G$ w.r.t.\ $S$ having size $O(n^\frac{4}{3} + |S| \cdot n \log n)$.
\end{corollary}

\noindent
This improves the additive stretch of the $(1,8)$-{\ttfamily SVS} of size $\widetilde{O}(n^\frac{4}{3})$ given in \cite{P14}, which holds for $|S|=\widetilde{O}(n^\frac{1}{3})$.

%As noted in the introduction, this allows us to build a $(1,4)$-\texttt{SES} of size $\widetilde{O}(n^{\frac{7}{5}}+ |S| \cdot n )$ by using the $(1,4)$-spanner of size $\widetilde{O}(n^\frac{7}{5})$ given in \cite{Che13}. This improves the $(1,4)$-\texttt{SES} of \cite{PP14} as soon as $|S| = \widetilde{\omega}(n^\frac{1}{15})$. Moreover, we can also improve the additive stretch of the $(1,8)$-{\ttfamily SVS} of size $\widetilde{O}(n^\frac{4}{3})$ given in \cite{P14}, which holds for $|S|=\widetilde{O}(n^\frac{1}{3})$, by using the $(1,6)$-spanner of size $O(n^{4/3})$ provided in \cite{BKMP10} to build a $(1,6)$-{\ttfamily SVS} of size $O(n^\frac{4}{3} + |S| \cdot n \log n)$.

\section{Conclusions}
\label{sct:conclusions}

In this paper, we have studied the problem of designing single-edge/vertex-fault-tolerant structures rooted at a source vertex, aiming at finding a compact set of edges of the input (either weighted or unweighted) graph that will provide approximate shortest paths from the source following the failure of an edge/vertex in the graph. The main contribution of our research is that we can get almost shortest paths with almost
linear size, in sharp contrast with a corresponding true-shortest paths structure which may require a quadratic size.
%all-to-all case, where a super-linear size is believed to be needed already for a single-edge failure and independently of the guaranteed stretch factor.
Another interesting contribution we provided is the bridging between $(\alpha,\beta)$-spanners and $(\alpha,\beta)$-{\ttfamily E/VABFS}.

The problem of designing good fault-tolerant approximate-shortest-path structures
deserves further investigation. For the single-source case, we mention three intriguing problems:
(1) designing a {\ttfamily SDSO} with stretch arbitrary close to $1$, almost linear size and constant query time for both the single-edge and the single-vertex failure scenario. The closest result is the {\ttfamily SDSO} given in \cite{BGLP16ESA} that has a logarithmic query time (w.r.t.\ the number of vertices of the graph) and only works for single edge failures; (2) removing the log-factor from the size of our structure, either improving its analysis or by further sparsifying it; (3) studying the multiple vertex-failure case. To the best of our knowledge there are no non-trivial {\ttfamily VASPT}s or {\ttfamily SDSO}s for this case.
Other future directions involve the study of the multisource case (i.e., a sourcewise fault-tolerant spanner), with the goal of designing a structure which only adds a sublinear (in the number of sources) term to the size of our single-source structure. Moreover, we also plan to investigate the existence of efficient fault-tolerant structures for other notable network topologies, like the minimum spanning tree, the tree spanner, or the minimum-routing cost spanning tree.

\end{document}